\documentclass[conference]{IEEEtran}
\IEEEoverridecommandlockouts
\usepackage{amsmath,amssymb,amsfonts, amsthm}
\usepackage{enumitem}
\usepackage{subfigure}
\usepackage{graphicx}
\usepackage{algorithmic}
\usepackage{graphicx}
\usepackage{subcaption}
\usepackage{bm}
\usepackage{textcomp}
\usepackage{xcolor}
\usepackage{hyperref}

\usepackage{fancyhdr}

\def\BibTeX{{\rm B\kern-.05em{\sc i\kern-.025em b}\kern-.08em
    T\kern-.1667em\lower.7ex\hbox{E}\kern-.125emX}}

\newcommand{\ket}[1]{\lvert #1 \rangle}
\newcommand{\bra}[1]{\langle #1 \rvert }

\newcommand{\matr}[1]{\boldsymbol{#1}}
\newcommand{\vect}[1]{\mathbf{#1}}
\newcommand{\svect}[1]{\boldsymbol{#1}}

\newcommand{\density}[1]{\mathcal{S}_+^{#1(1)}}

\newcommand{\mathbbm}[1]{\text{\usefont{U}{bbm}{m}{n}#1}}

\def\eqref#1{equation~\ref{#1}}

\def\1{\boldsymbol{1}}

\DeclareMathAlphabet{\mathsfit}{\encodingdefault}{\sfdefault}{m}{sl}
\SetMathAlphabet{\mathsfit}{bold}{\encodingdefault}{\sfdefault}{bx}{n}

\usepackage[natbib=true, style=ieee, backend=biber, maxbibnames=9, maxcitenames=2, mincitenames=1, sorting=none]{biblatex}
\theoremstyle{plain}
\newtheorem{lemma}{Lemma}
\newtheorem{theorem}{Theorem}
\newtheorem{corollary}{Corollary}
\theoremstyle{definition}

\theoremstyle{remark}
\newtheorem{remark}{Remark}
\newtheorem{definition}{Definition}

\begin{document}

\title{Certifiably Robust Encoding Schemes\\
\thanks{The project/research is supported by the Bavarian Ministry of Economic Affairs, Regional Development and Energy with funds from the Hightech Agenda Bayern.}
}

\author{
    \IEEEauthorblockN{Aman Saxena\IEEEauthorrefmark{1}\IEEEauthorrefmark{3}, Tom Wollschläger\IEEEauthorrefmark{1}\IEEEauthorrefmark{3}, Nicola Franco\IEEEauthorrefmark{2}, Jeanette Miriam Lorenz\IEEEauthorrefmark{2}, Stephan G\"unnemann\IEEEauthorrefmark{3}}
    \IEEEauthorblockA{\IEEEauthorrefmark{3}School of Computation, Information \& Technology, Technical Univ. of Munich, Germany
    \\\{a.saxena, t.wollschlaeger, s.guennemann\}@tum.de}
    \IEEEauthorblockA{\IEEEauthorrefmark{2}Fraunhofer Institute for Cognitive Systems IKS, Munich, Germany
    \\\{nicola.franco, jeanette.miriam.lorenz\}@iks.fraunhofer.de}
    \IEEEauthorblockA{\IEEEauthorrefmark{1} denotes equal contribution}
}

\fancypagestyle{specialfooter}{%
  \fancyhf{}
  \renewcommand\headrulewidth{0pt}
  \fancyfoot[R]{ \noindent\fbox{%
    \parbox{\textwidth}{%
        {\footnotesize \copyright 2024 IEEE. Personal use of this material is permitted. Permission from IEEE must be obtained for all other uses, in any current or future media, including reprinting/republishing this material for advertising or promotional purposes, creating new collective works, for resale or redistribution to servers or lists, or reuse of any copyrighted component of this work in other works.}
        }
    }}
}

\maketitle
\thispagestyle{plain}
\pagestyle{plain}
\thispagestyle{specialfooter}

\begin{abstract}

Quantum machine learning uses principles from quantum mechanics to process data, offering potential advances in speed and performance. However,  previous work has shown that these models are susceptible to attacks that manipulate input data or exploit noise in quantum circuits. Following this, various studies have explored the robustness of these models. These works focus on the robustness certification of manipulations of the quantum states. 
We extend this line of research by investigating the robustness against perturbations in the classical data for a general class of data encoding schemes. We show that for such schemes, the addition of suitable noise channels is equivalent to evaluating the mean value of the noiseless classifier at the smoothed data, akin to Randomized Smoothing from classical machine learning.  
Using our general framework, we show that suitable additions of phase-damping noise channels improve empirical and provable robustness for the considered class of encoding schemes. 
\end{abstract}


\section{Introduction}

In recent years, the increasing success in quantum computing has propelled the interest in Quantum Machine Learning (QML) with the hope of advances in speed and performance. QML models typically encode the data into a quantum state and then use a variational quantum circuit (VQC) to transform those states before a measurement returns the prediction. However, research has shown that QML models are vulnerable to adversarial perturbations \cite{Quantum_ADV_Liu_2020, Quantum_ADV_Lu_2020}, which means that adversaries can attack the model to change its prediction. \citet{haar_random_samples, robust_in_practice} give theoretical arguments based on the concentration of measure exposing the vulnerability of QML models against adversarial noise. In the machine learning image domain, they add specifically crafted noise to an image, resulting in a change that is not visible to the human eye but can drastically alter the model's prediction. For QML, the adversary can manipulate either the classical data or the quantum states acting as input to the VQC. Moreover, the noise present in the quantum circuit during the data encoding process could also act as an adversary.

Several studies proposed ways to improve robustness or develop measurable guarantees to respond to the susceptibility of QML approaches \cite{QHT_robustness, gong2022enhancing, montalbano2024quantumadversariallearningkernel}. A notable research direction involves adding noise to derive differential privacy guarantees and extending them to provable robustness, \cite{quantum_noise_protects_2, quantum_noise_protects}. These methods model the perturbations within the space of the quantum state. However, while using QML on classical data, an adversary can create a strong attack utilizing the data encoding.

As a new approach, our work certifies models against perturbations in the classical data, enabling the certificate for the entire data processing pipeline. We introduce the incorporation of appropriate quantum noise channels and provide efficient certification of QML models against perturbations in classical data. Therefore, while aligned with related work, our approach stands out from prior research by addressing challenges posed by classical adversaries and offering efficient solutions for their certification. Particularly, this line of research is entirely different from using quantum computing to speed up the existing certification algorithm for classical machine learning models \cite{franco2022quantum, franco2023efficient, DRSOurs}.

For a general class of encoding schemes as proposed in \cite{schuld_effect_of_de}, we draw inspiration from the concept of randomized smoothing \cite{cohen2019}, a stochastic certification technique from classical machine learning. Hence, we smooth our data by adding noise and can thus derive guarantees for the smoothed version. We construct \textit{smooth} encoding schemes based on existing ones that efficiently certify the robustness of quantum classifiers. Specifically, for encoding schemes that use single-qubit rotation gates to encode the data, we show that adding suitable phase-damping or dephasing channels is equivalent to smoothing some classifiers as proposed by \citet{cohen2019}. We evaluate our approach on one artificial and two real-world datasets, \textit{MNIST} and \textit{TwoMoons}. We show that the amount of noise we use for smoothing drastically affects both the clean and the certified accuracy. We further analyze the certification strengths using an adversarial attack to create an empirical upper bound on the provable robustness.

\noindent Our main contributions are summarized as follows:
\begin{itemize}
   \item We develop a new randomized smoothing-inspired smooth encoding scheme that provides improved robustness and measurable guarantees.
    \item We show that our smooth encoding is equivalent to using suitable phase-damping channels for specific encoding schemes.
    \item We demonstrate the effectiveness of our method on multiple different datasets and analyze the strength of the certificate by providing an upper bound through adversarial attacks\footnote{Find our code at \href{https://www.cs.cit.tum.de/daml/robust-encoding-schemes}{cs.cit.tum.de/daml/robust-encoding-schemes}.}.
\end{itemize}

\label{sec:intro}
\section{Background}
\noindent\textbf{Notations and Definitions:} We define the following notations and briefly review some concepts that will be useful for our subsequent discussion:
\begin{itemize}
    \item $\eta(y):=  \mathbbm {1} \{y > \frac{1}{2}\}$.
    \item For $\vect x \in \mathcal X$ we define $\mathcal{B}_{\vect D, \epsilon}(\vect {x}) := \{\vect z \in \mathcal X : \vect D(\vect x, \vect z) \le \epsilon \}$ as the ball of radius $\epsilon$ with respect to the metric $\vect D$ defined in the data space $\mathcal X$.
    \item $\mathcal{S}_+^{d(1)} := \{\matr{\rho} \in \mathbb{C}^{2^d \times 2^d} | \rho \succeq 0 \land \text{Tr}(\matr{\rho}) = 1\}$; the space of Positive Semi-Definite (PSD) matrices of dimension $2^d \times 2^d$ with trace $1$.
    \item $\Phi(x):= \int_{-\infty}^{x} \frac{1}{\sqrt{2 \pi}} \exp(\frac{-z^2}{2}) dz$; the Cumulative Density Function (CDF) of the standard Gaussian distribution.x
    \item Unless specified otherwise, $\matr {\delta} \sim \phi$ implies that $\phi$ is the characteristic function of the distribution.
    \item POVM (Positive Operator Valued Measure): For the scope of this paper, POVM is a set of Hermitian, positive semi-definite operators $\{\matr {\Pi_i}\}_{i=1}^{N}$ such that $\sum_{i=1}^{N} \matr {\Pi_i} = \mathbb I$. POVMs are extremely useful for defining the probability measure on the outcome of the quantum measurement. Specifically, given the quantum state $\rho \in \mathcal{S}_+^{d(1)}$, the probability of observing the $i^{th}$ outcome is given as $p(i) = Tr(\matr {\Pi_i} \rho)$.
    \item CPTP (Completely Positive Trace Preserving) map defines the time evolution of open quantum systems, i.e. the systems that can be entangled with the environment. A map $\mathcal E: \mathcal{S}_+^{d(1)} \mapsto \mathcal{S}_+^{d(1)}$ is a CPTP map if and only if $\exists \{E_i\}_{i=1}^{M}$ such that $\sum_{i=1}^{M} E_i^H E_i = \mathbb I$ and $\mathcal E(\rho) = \sum_{i=1}^{M} E_i \rho E_i^H$. 
    \item Trace distance between two states $\matr{\rho_1}$ and $\matr{\rho_2}$ is the Shatten-1 norm of $\matr{\rho_1} - \matr{\rho_2}$. Let $\{\sigma_i(\matr A)\}_{i=1}^N$ be the singular values of $\matr A$, then $Tr(\matr{\rho_1}, \matr{\rho_2}) = \lvert \lvert \matr{\rho_1} - \matr{\rho_2} \rvert \rvert_*:= \sum\sigma_i(\matr{\rho_1} -  \matr{\rho_2})$.\\
\end{itemize}

\noindent\textbf{Quantum Classifiers:} In our work we focus on supervised quantum classifiers. The core idea of these models is to use a parameterized quantum circuit as a machine-learning model. The probability of predicting a certain class given a data point is associated with a certain outcome upon measuring an observable. The task of learning is to find the optimal parameters minimizing a given loss function. If we work with classical data we have a two-stage process. In the first part, we encode the classical data into the corresponding quantum states by applying a quantum circuit parameterized in the data. In the second stage, we try to find an optimal measurement that minimizes the desired loss function. These models are similar to kernel methods. Important properties of such models have been extensively discussed in related work \cite{schuld_effect_of_de, schuld_kernel_methods, expo_growing}. We keep our discussion focused on binary classifiers which can easily be extended to multi-class classifiers. We formally define the quantum binary classifiers for classical data as follows:

\begin{definition}[Quantum Binary Classifier]
    Let $\vect x \in \mathbb R^n$ be our data and $\rho_{\theta_1}: \mathbb R^n \mapsto \mathcal{S}_+^{d(1)}$ be the parameterized data encoding circuit. In addition, $\mathcal{E}_{\theta_2}(.)$ is a parameterized CPTP map and $\{\matr {\Pi},  \mathbb I - \matr {\Pi}\}$ a POVM, defining the measurement process. The classifier $y: \mathbb R^n \mapsto [0,1]$ then is:   
    \begin{align*}
        y(\vect x; \theta_1, \theta_2) := Tr(\matr{\Pi} \mathcal{E}_{\theta_2}(\matr{\rho_{\theta_1}(\vect x)}))
    \end{align*}
    Further $f(\vect x; \theta_1, \theta_2) := \eta(y(\vect x; \theta_1, \theta_2)) \in \{0,1\}$ returns the class prediction. \\
    \label{def:qnn}
\end{definition}

\noindent\textbf{General near-term Encoding Schemes:} \citet{schuld_effect_of_de} propose a general near-term encoding scheme and show that they represent partial Fourier series in the features. The frequency spectrum that the underlying models represent depends on the number of data encoding gates. In addition, the authors prove the universal approximation theorem for the increasingly extensive frequency spectrum and argue that they can represent various encoding schemes from the literature with their framework upon classically transforming the features. In particular, the classically intractable encoding scheme of \citet{quantum_enhanced_fs} can be represented using the framework of \citet{schuld_effect_of_de} using feature transformations of the form $\phi_{ij}(\vect x) = x_i x_j$.
\begin{figure}[t]
    \centering
    \includegraphics[width=1\linewidth]{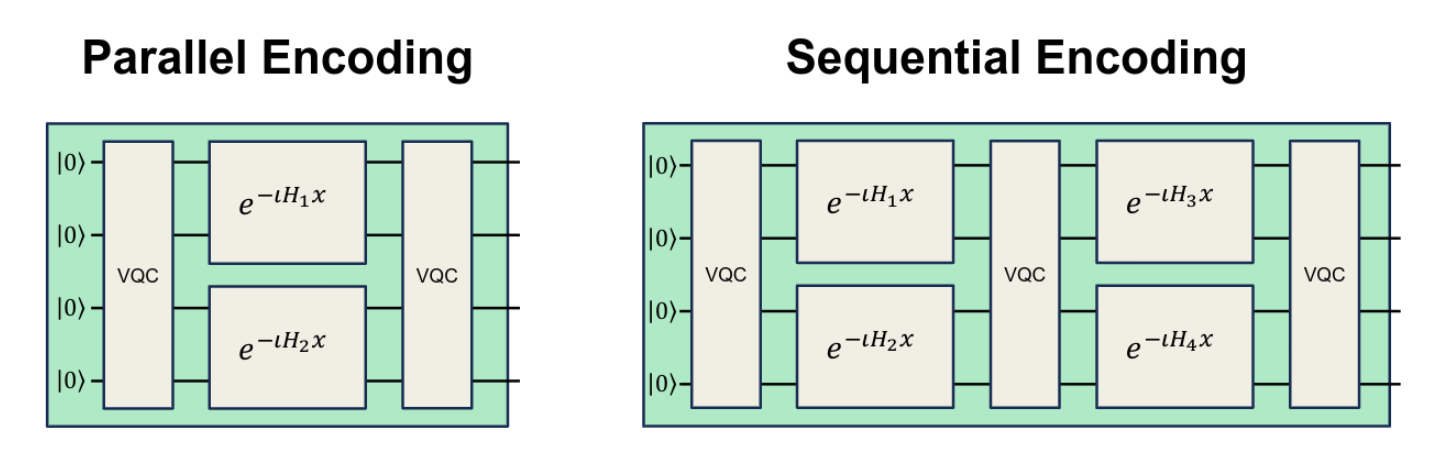}
    \caption{Parallel Encoding maps data to parallel quantum subsystems while sequential encoding is just a stack of parallel layers.} 
    \label{2:fig:qml:parallelvsseq}
\end{figure}
Typically these schemes encode data with the eigenvalues of the Hamiltonian governing the system. 

We categorize the QML models discussed here into two categories depending on how the data is encoded: \textbf{parallel} or \textbf{sequential}.

\subsubsection{Parallel Encoding Schemes}\label{2:sec:qml:parallel}

In the case of parallel schemes, data is encoded in parallel subsystems using their respective Hamiltonians. 
Generally, these Hamiltonians are assumed to be diagonal\footnote{Otherwise use the spectral decomposition to write them as $V \Lambda(x) V^H$ and assume $V, V^H$ to be a part of the immediate VQC gates.}. 
We denote them by $\Lambda(x):= \text{diag}(\exp(-\iota \lambda_i \vect x))$. The density matrix after the encoding block can then be written as $\rho(\vect x) = \Lambda(\vect x)\ket{\gamma} \bra{\gamma} \Lambda(\vect x)^H$. Thus, the individual component is given by:
\begin{equation}
\label{2:eq:qml:parallel}
        \matr \rho(\vect x)_{ij} = \gamma_i  \gamma_j^* \exp(-\iota(\lambda_i - \lambda_j) \vect x).    
    \end{equation}  

\subsubsection{Sequential Encoding Schemes}
\label{2:sec:qml:sequential}
Another way to extend the set of frequencies is by repeatedly inserting the Hamiltonian encoding layers sequentially. A sequential layer is essentially a repetition of parallel encoding blocks as follows: 

\begin{align}
\label{2:eq:qml:sequential}
    \matr \rho(\vect x) := &U_L(\vect x)W^{(L)}...U_1(\vect x)W^{(1)}\rho_0 W^{(1)H}U_1(\vect x)^H... \\
    &W^{(L)H}U_L(\vect x)^H. 
\end{align}
where $U_l(x) := diag([\exp(-\iota \lambda^{(l)}_i x_i)])$ represents the diagonal data encoding block.\\

\noindent\textbf{Adversarial Robustness:}
Given an instance of data $\vect {x_0}$, adversarially attacking the classifier $f(\cdot)$ at $\vect{x_0}$ typically involves finding a \textit{semantically} similar data point $\vect x$ such that $f(\vect x) \neq f(\vect {x_0})$. Typically, two data points are said to be semantically similar if they are close enough with respect to some metric defined on the input space. We define this metric-based notion of robustness.
 \begin{definition}[$(\vect D, \epsilon)$-Robust]
    Let $f : \mathcal{X} \mapsto \{0,1\}$ be the classifier on the data space $\mathcal{X}$, then $f$ is said to be $(\vect D, \epsilon)$-Robust at $\vect{x_0}$ if the following holds:
    \begin{align*}
        f(\vect{x}) = f(\vect{x_0}) \quad \forall \vect{x} \in \mathcal{B}_{\vect D, \epsilon}(\vect {x_0}).
    \end{align*}
    \label{def:adv_rob}
\end{definition}
In the continuous classical machine learning setting, the data space $\mathcal{X}$ could be some Euclidean space $\mathbb R^n$, and $\vect D$ induced from $l_p$-norm $\lvert\lvert \cdot \rvert\rvert_p$.
Providing robustness guarantees typically involves addressing complex optimization problems that can be non-linear, non-convex, and highly dependent on the model's structure. Given a quantum classifier $f(\vect x) := \eta(Tr(\matr{\Pi} \mathcal{E}(\matr{\rho(\vect x)})))$, the certification problem at $\vect {x_0} \in \mathbb R^n$ can be reduced to finding the worst-case margin as follows:
\begin{equation}
    \label{eq:4:opt_class}
    \begin{aligned}
    \vect {m^*} := &\min_{\vect x \in \mathbb R ^n} (-1)^{f( \vect{x_0})} (1 - 2Tr(\matr{\Pi} \mathcal{E}(\matr{\rho(\vect x)})))\\
    &\text{subject to} \quad \lvert\lvert \vect x - \vect {x_0} \rvert \rvert_p \leq \epsilon .
    \end{aligned}
\end{equation}

This margin encodes whether or not the adversary can fool the model for any of the admissible perturbed data points. If we can find the lower bound $\vect {\underline{m^*}}$ of $\vect {m^*}$ such that $\vect {\underline{m^*}} > 0$, then the classifier $f$ is certifiably $(\lvert \lvert \cdot \rvert \rvert_p, \epsilon)$-Robust at $\vect{x_0}$. One strategy involves identifying and optimizing a convex lower bound for the objective function \cite{wong2018provable}. Alternatively, randomized smoothing can be used to stochastically obtain a monotonic lower bound for the smooth classifier, thus simplifying the optimization challenge \citet{Zhang2020BlackBoxCW, cohen2019}.

\label{sec:background}

\section{Related Works}
The existing literature predominantly addresses the robustness certification of quantum classifiers against perturbations in quantum states. Given an input quantum state $\rho \in \density(d)$, related work establishes $(D, \epsilon)$-Robustness guarantees for the underlying classifier. Typically, the metric $D$ is selected as the trace distance on density matrices or any other metric induced by fidelity distance. Drawing inspiration from the classical RS framework \cite{cohen2019}, recent works \cite{quantum_noise_protects, quantum_noise_protects_2} propose to use noise to the classifier's operations enabling the derivation of robustness guarantees against perturbations measured in the trace distance. Specifically, \citet{quantum_noise_protects} use the addition of a depolarization noise channel, thereby probabilistically transforming the input state to the maximally mixed state with a probability $p$. They derive robustness guarantees by first establishing the differential privacy of the noisy classifier. However, as demonstrated in \citet{QHT_robustness}, these guarantees prove to be weaker even in comparison to the trivial certification obtained using Hölder's duality.
Moreover, their practical application is limited when dealing with higher-dimensional datasets~\cite{winderl2023quantum}.
Building on this line of research, \citet{winderl2024constructing} expanded the scope of $(D, \epsilon)$-Robustness guarantees, previously observed for depolarizing and random rotation channels, to arbitrary noise channels.
In contrast, \citet{QHT_robustness} addresses the certification problem through QHT and achieves guarantees that surpass trivial bounds. 

Despite certifying perturbations in the input quantum state, one does not receive any guarantees for perturbations in the underlying input classical data point. For example, to certify at $\vect{x_0}$, we still need to solve the following optimization problem:

\begin{equation}
    \label{eq:4:opt_class_to_quant}
    \begin{aligned}
    \vect {d^*} := &\max_{\vect{x} \in \mathbb R^n} \vect D(\matr \rho(\vect x), \matr \rho(\vect {x_0}))\\
    &\text{subject to} \lvert\lvert \vect{x} - \vect{x_0} \rvert\rvert \leq \epsilon. 
    \end{aligned}
\end{equation}

Here, $\vect D$ represents the desired metric on the quantum states. To certify perturbations in the classical data space, one must identify the non-trivial upper bound of $\vect{d^*}$ and verify whether it is less than the maximum allowed perturbation in the quantum space derived from previous work.

In our work, we focus on deriving such guarantees for general near-term encoding schemes \cite{schuld_effect_of_de}. We use the randomized smoothing technique to find the convex upper bound of the objective function in \autoref{eq:4:opt_class_to_quant}. Furthermore, we show that such smoothing can be realized using suitable noise channels. Essentially, we derive noise channels that make our encoding scheme and hence the whole quantum classifier efficiently certifiable.   

\label{sec:relatedwork}

\section{Methods}
\subsection{Randomized Smoothing as Quantum Noise}
     
Randomized smoothing (RS), as introduced in \cite{cohen2019}, constructs a smooth classifier from a base classifier $f$. This is achieved by adding Gaussian noise to the input and then calculating the classifier's mean prediction for different noisy samples. Specifically, the smooth classifier $g(\vect x)$ is evaluated as $\mathbb E_{\matr \delta \sim \mathcal N(0, \sigma^2)} [f(\vect x + \matr \delta)]$. In our work, we draw inspiration from RS and change an existing encoding scheme $\matr \rho(\vect x)$ to a smoothed version $\matr{\tilde \rho} (\vect x)$. Instead of directly encoding $\vect x \in \mathbb R^n$, we sample noise $\matr \delta $ from a zero-centered distribution $\phi$ with finite variance and encode $\vect x + \matr \delta$ instead. We then take the expected value of the random variable matrix $\matr \rho(\vect x + \matr \delta)$ to obtain the smooth encoding of the data point $\vect x$, i.e., $\matr{\tilde \rho} (\vect x) := \mathbb E_{\matr \delta \sim \phi}[\matr \rho(\vect x+ \matr \delta)] \in \density{d}$. We derive a black-box upper bound to the bound that is independent of the base encoding scheme, see Equation \ref{eq:4:opt_class_to_quant}.

We first consider a simple case of parallel encoding schemes, for which we show that the smoothing operation with respect to any given distribution $\phi$ can be realized as a data-independent quantum channel\footnote{data-dependent channels if nonlinear feature transformations are used}. All results on smooth parallel encoding schemes are formalized in Theorem \ref{Th:Parallel_Encoding}, \ref{Th:Paralle_Upper_Bound}. We generalize these results to any Hamiltonian-based encoding scheme.%
The idea is to apply appropriate quantum noise to each parallel layer. The smooth sequential encoding scheme is formally defined in Definition \ref{def:smooth_sequential_encoding} and the formal guarantees are derived in Theorem \ref{Th:Sequential_Encoding}.
Further, we show that smoothing the single-qubit Hamiltonian encoding is equivalent to applying appropriate phase-damping (dephasing) channels. Any parallel Hamiltonian encoding that encodes data into single-qubit subsystems is the same as a stack of single-qubit parallel encodings. Therefore, we can use our smooth sequential encoding formalism to derive guarantees after applying appropriate phase-damping channels.  

We present results tailored for 1D data, specifically when $\mathcal{X} := \mathbb{R}$. These insights can be readily extended to multiple dimensions, as smooth sequential encoding inherently accomplishes. We defer all proofs are deferred to the appendix.\\

\noindent \textbf{Smooth Parallel Encoding Schemes:}
Assume, $\matr \delta \sim \phi$, using the representation of the parallel encoding schemes from Equation \ref{2:eq:qml:parallel}, we can express the smooth version as:
\begin{align*}
 \matr{\tilde \rho}(\vect x)_{ij}: &= \mathbb E_{\matr \delta \sim \phi}[\matr \rho(\vect x+ \matr \delta)_{ij}] \\ 
 &=\mathbb E_{\matr \delta} [\gamma_i \gamma_j^*\exp(-\iota (\lambda_i - \lambda_j)(\vect x + \matr \delta))] \\
 &= \gamma_i \gamma_j^*\exp(-\iota (\lambda_i - \lambda_j)\vect x) \phi(\lambda_j - \lambda_i)\\
 &= (\matr \rho \odot A_{\phi, \lambda})_{ij},
\end{align*}
where $A_{\phi, \lambda} := [\phi(\lambda_j - \lambda_i)]_{ij}$\footnote{e.g., for Gaussian distribution $A_{\phi, \lambda} := [\exp(\frac{\sigma^2(\lambda_j - \lambda_i)^2}{2})]_{ij}.$} and $\odot$ denotes point-wise multiplication. The transformation $\rho \mapsto \rho \odot A_{\phi, \lambda}$ is linear on the space of density matrices. In the following theorem, we show that this smoothing transformation is a valid quantum operation and derive its Kraus representation. 
\begin{theorem}[Smooth Parallel Encoding: Smoothing is a CPTP map]
    Let $\matr \delta \sim \phi$ such that $\mathbb E_{\matr \delta}[\matr \delta] = 0$ and $E_{\matr \delta}[\matr \delta^2] < \infty$. Define $\matr{\rho}(\vect x) := \matr U(\vect x)\ket{\gamma}\bra{\gamma}\matr U(\vect x)^H$ where $\matr U(\vect x) := diag([\exp(-\iota \lambda_i \vect x]_i)$, $\matr A_{\phi, \lambda} := [\phi(\lambda_j - \lambda_i)]_{ij}$ and $\mathcal E_\phi:  \matr \rho \mapsto \matr \rho \odot \matr{A}_{\phi, \lambda}$, the following holds: 
    \begin{enumerate}[label=(\alph*)]
        \item $\mathcal E_\phi$ is a CPTP map and $\mathbb E_\delta[\matr{\rho}(\vect x+ \matr \delta)] = \mathcal E_\phi (\matr{ \rho}(\vect x))$.

        \item $\matr A_{\phi, \lambda} \succeq 0$, let $\matr A_{\phi, \lambda} = \sum_k \sigma_k \vect u_k \vect u_k ^ H$ be its spectral decomposition and define $\matr {E_k} := \sqrt{\sigma_k} diag(u_k)$, then $\{\matr {E_k}\}_k$ represents the Kraus representation of $\mathcal E_\phi$, i.e.  $\mathcal E _\phi(\matr \rho) = \sum_k \matr {E_k} \matr \rho \matr {E_k} ^H $ and $\sum_k \matr {E_k}^H \matr {E_k} = \mathbb I $.        \label{subth:trace_distance}
    \end{enumerate}
    \label{Th:Parallel_Encoding}
\end{theorem}
Application of the smooth quantum operation $\mathcal E_\phi$ defined above, allows us to derive a black box and non-trivial upper bound of the trace distance between the smooth encodings of the two data points. We state this result as Theorem \ref{Th:Paralle_Upper_Bound}. Obtaining such bounds can make the optimization problem in Equation \ref{eq:4:opt_class_to_quant} tractable, thus facilitating the efficient certification of these QML models. This is evident in 
 Corollary~\ref{cor:Parallel_Gaussian} showcasing the results for Gaussian smoothing where we get an upper bound of the trace distance that is monotonic with respect to the $l_2$ norm of the input data points. 

\begin{theorem}[Smooth Parallel Encoding: Upper bound of the Trace distance]
Following the notation of Theorem \ref{Th:Parallel_Encoding} but we additionally assume that the distribution $\phi$ has a probability distribution function $f_\phi$ and denote $d(\cdot,\cdot)$ as the trace distance between the density matrices. Then the distance is bounded by: 
    \begin{multline}
        \label{eq:above_bound_1}
        d(\mathcal E_\phi(\matr \rho(\vect x)), \mathcal E_\phi(\matr \rho(\vect y))) \leq \\
        \int_{\mathbb R} \max(f_\phi(\vect z- \vect x) - f_\phi(\vect z- \vect y), 0) d\lambda(\vect z).
    \end{multline}
    \label{Th:Paralle_Upper_Bound}
\end{theorem}
We apply Theorem \ref{Th:Paralle_Upper_Bound} for the case of Gaussian smoothing and express the results as Corollary \ref{cor:Parallel_Gaussian}. The first part gives an upper bound for the trace distance while the second part establishes robustness guarantees.  
    \begin{corollary}[Gaussian parallel smooth encoding scheme: Infinite sampling] Following the setup in Theorem \ref{Th:Parallel_Encoding}
    and using a Gaussian distribution with variance $\sigma^2$ as a smoothing distribution, i.e. $\matr \delta \sim \mathcal N(0, \sigma^2)$, we define the quantum classifier using the parallel encoding $\matr \rho(\vect x)$ as $f(\vect x) := Tr(\matr P\mathcal E (\mathcal{E}_\phi(\matr \rho(\vect x))))$. Then the following holds:  
        \begin{enumerate}[label=(\alph*)]
        \item  The upper bound on the trace distance between the perturbed quantum smooth sates is given as            
            \begin{equation}
                \label{eq:above_bound_2}
                d(\mathcal E_\phi(\matr \rho(\vect x)), \mathcal E_\phi(\matr \rho(\vect y))) \leq 2 \Phi (\frac{\lvert \vect x - \vect y \rvert}{2\sigma}) - 1. 
            \end{equation}

        \item If $f(\vect {x_0}) \geq f_0 > \frac{1}{2}$ then, $f(\vect {x_0} + \matr \delta) > \frac{1}{2} \forall \matr \delta $ for which holds that $\lvert \matr \delta \rvert < \sigma \Phi^{-1}(f_0)$. 

    \end{enumerate}
    \label{cor:Parallel_Gaussian}
    \end{corollary}
Let a measurable transformation $f$ be applied to the features before encoding them using parallel encoding, that is, $\matr{\tilde \rho}(f(\vect x))$ is used as an encoding block. For $\matr \delta \sim \phi$, define a random variable $Y_{\vect x} = f(\matr \delta + \vect x) \sim \phi_{\vect x}$. The corresponding smooth encoding scheme can be written as $\matr {\rho_\phi}(\vect x) = \mathbb E_{\matr \delta}[\matr {\rho}(f(\vect x + \matr \delta))] = \ket{\gamma}\bra{\gamma} \odot [\phi_{\vect x}(\lambda_j - \lambda_i)]_{ij}$. Therefore, using Theorem \ref{Th:Parallel_Encoding} we can show that this encoding map is CPTP, and hence the overall encoding map can be realized as a quantum operation.\\

\noindent \textbf{Smooth Sequential Encoding Schemes:}
The smoothing channels discussed in Theorem \ref{Th:Parallel_Encoding} involves finding the spectral decomposition of the exponentially growing system $A_{\phi, \lambda}$\footnote{Finding the Kraus representation of the smoothing operation involves finding spectral decomposition of $A_{\phi, \lambda}$; see part (b) of Theorem \ref{Th:Parallel_Encoding}} followed by decomposing the resulting operation into the fundamental gates. Therefore, it can be efficient to view a complex parallel encoding block as a stack of simple parallel encoding blocks where encoding gates are only applied to a fixed number of qubits. For example, view a $d$ qubit parallel encoding block as a stack of $d$ single qubit parallel encoding gates as shown in Figure \ref{2:fig:qml:parallelvsseq}. %
Therefore, for sequential encoding schemes (with L parallel layers), we propose to apply $\mathcal E_\phi$ from Theorem \ref{2:eq:qml:parallel} to each of the parallel layers and then interpret the resulting smooth encoding scheme as smoothing a data encoding scheme for $L$-dimensional data\footnote{For $D$-dimensional data, we interpret sequential smoothing encoding as smoothing $LD$-dimensional data. Given that even for $1$-D data, we operate with multidimensional data for smooth sequential encoding, this framework seamlessly extends to multidimensional data}. We define a smooth sequential encoding in the following definition and formalize above arguments in Theorem \ref{Th:Sequential_Encoding}.

  \begin{definition}[Smooth Sequential Encoding]
    Let $\matr \delta \sim \phi$ and $\matr \rho(\vect x)$ be the general sequential encoding as defined in Equation \ref{2:sec:qml:sequential}. For $A_{\phi, \lambda^{(l)}} := [\phi(\lambda^{(l)}_j - \lambda^{(l)}_i)]_{ij}$ let $\mathcal E^{l}_\phi:  \rho \mapsto \rho \odot \matr{A}_{\phi, \lambda^{(l)}}$ define the smoothing quantum channel for the $l^{th}$ parallel layer, then the smooth sequential encoding scheme is: 
    \begin{multline*}
        \matr{\tilde \rho}_{L, \phi}(x) := \mathcal E^L_\phi(U_L(x)W^{(L)}... \\ \mathcal E^1_\phi(U_1(x)W^{(1)}\rho_0W^{(1)H}U_1(x)^H)...W^{(L)H}U_L(x)^H).
    \end{multline*}  \label{def:smooth_sequential_encoding}
  \end{definition}
We now present our \textbf{main} result in the direction of certifiable encoding schemes followed by its realization for Gaussian smoothing in Corollary \ref{cor:sequential_gaussian}.
\begin{theorem}[Smooth Sequential Encoding: Infinite Sampling]
    
    Let $\matr{\tilde \rho}_{L, \phi}$ denote the general smooth sequential encoding scheme from Definition \ref{def:smooth_sequential_encoding}. Further, for $\vect{z}:=[z_1, ..., z_L]$, define $\sigma_L(\vect{z}) := \matr {U_L}(z_L)\matr {W^{(L)}}...\matr {U_1}(z_1)\matr {W^{(1)}}\matr {\rho_0}\matr W^{\bm{( 1)}H}\matr {U_1}(z_1)^H \\ ...\matr W^{{\vect(L)}H}\matr {U_L}(z_L)^H$, then the following holds:
    
    \begin{enumerate}[label=(\alph*)]
        \item Let $\vect{x^{(L)}} = [x, x, ..., x] \in \mathbb R^L$, and $\svect{\delta^{(L)}}=[\delta_1, \delta_2, ..., \delta_L]$ be a random vector of iid. $\delta_i \sim \phi$, then $\matr{\tilde \rho}_{L, \phi}(\vect x) = \mathbb E_{\svect{\delta^{(L)}}}[\sigma_L(\vect{x^{(L)} }+ \matr \delta^{(L)})]$. 

        \item Additionally assume that the distribution $\phi$ has a probability distribution function $f_\phi$ and denote $d(.,.)$ as the trace distance between the density matrices then, 
            \begin{multline}                
                \label{eq:above_bound_5}
                d(\matr{\tilde \rho}_{L, \phi}(\vect x), \matr{\tilde \rho}_{L, \phi}(\vect y)) \leq \\ \int_{\mathbb R^L} \max(\prod_{i=i}^L f_\phi(z_i-\vect x) - \prod_{i=i}^L f_\phi(z_i-\vect y), 0) d\lambda(\vect{z}).
            \end{multline}

    \end{enumerate}
    \label{Th:Sequential_Encoding}
\end{theorem}
\begin{corollary}
    [Gaussian sequential Smooth Encoding Scheme: Infinite Sampling]
    Following the setup in Theorem \ref{Th:Sequential_Encoding} and using the Gaussian distribution with variance $\sigma^2$ as the smoothing distribution, i.e. $\matr \delta \sim \mathcal N(0, \sigma^2)$, define the quantum classifier using the smooth sequential encoding $\matr {\tilde \rho}(\vect x)$ as $f(\vect x) := Tr(\matr P\mathcal E (\matr {\tilde \rho}(\vect x)))$. Then the following holds:
    \begin{enumerate}[label=(\alph*)]
    \item  The upper bound on the trace distance between the perturbed quantum smooth states is given as:            
        \begin{equation}
            \label{eq:above_bound_3}
            d(\matr{\tilde \rho}_L(\vect x), \matr{\tilde \rho}_L(\vect y)) \leq 2 \Phi (\frac{\sqrt{L}\lvert \vect x - \vect y \rvert}{2\sigma}) - 1.
        \end{equation}
    \item Let $f(\vect {x_0}) \geq f_0 > \frac{1}{2}$ then, $f(\vect {x_0} + \matr \delta) > \frac{1}{2} \quad \forall \matr \delta $; $\lvert \delta \rvert < \frac{\sigma}{\sqrt{L}} \Phi^{-1}(f_0)$. 
    \end{enumerate}
    \label{cor:sequential_gaussian}
\end{corollary}

\subsection{Smoothing single-qubit Pauli rotation gates}
Many encoding schemes in the literature \cite{expo_growing, quantum_enhanced_fs, schuld_effect_of_de} utilize single-qubit Pauli rotation gates, such as $RX$, $RY$, and $RZ$, to encode data. We demonstrate the applicability of Theorem~\ref{Th:Sequential_Encoding} to construct a smooth encoding scheme for such a single-qubit data encoding block. We detail the case where the data encoding block utilizes the Pauli-$Z$ rotation gate $RZ$; other cases easily follow. We establish in Theorem~\ref{th:singlesmooth} that the smoothing channel required to smooth $RZ$, according to Theorem \ref{Th:Parallel_Encoding}, is a phase damping channel.%
\begin{theorem}[Smooth $RZ$ data encoding block]
    \label{th:singlesmooth}
    Denote the single qubit subsystem with the $RZ$ encoding block as $\matr S$ and the remaining subsystem as $\matr R$. Let $RZ(\vect x) \otimes \mathbb I_{\matr R}$ be the data encoding circuit that acts on the system $d$ qubit system $\matr{SR}$. For a random variable $\matr \delta$ with real characteristic function $\phi$, we define $\lambda := 1 - (\phi(1))^2$. Let $\mathcal N_{PD}(\lambda)$ be a phase damping noise channel with parameter $\lambda$\footnote{For zero-centered Gaussian distribution with standard deviation $\sigma$, $\lambda = 1 - \exp(-\sigma^2)$.}. Then, $\forall \matr \rho \in \mathcal{S}_+^{d(1)}$:
    \begin{multline}
       \mathcal N_{PD}(RZ(\vect x) \otimes \mathbb I_{\matr R} \text{ } \matr \rho \text{ } RZ(\vect x)^H \otimes \mathbb I_{\matr R}) = \\ \mathbb{E}_{\matr \delta \sim \phi}[RZ(\vect x + \matr \delta) \otimes \mathbb I_{\matr R} \text{ } \matr \rho \text{ } RZ(\vect x + \matr \delta)^H \otimes \mathbb I_{\matr R}].         
    \end{multline}
\end{theorem}

\subsection{Additional Remarks}
We end this section with important remarks concerning smooth $RY$, $RZ$ data encoding blocks, and single-qubit data encoding blocks employing nonlinear transformations to the features prior to encoding.

\paragraph{Smoothing $RY$ and $RZ$ data encoding blocks} Using spectral decomposition of $RX$ and $RY$ gates they can be written as $R = V (\matr{RZ}) V^H $. Let, the Kraus representation of $\mathcal N_{PD}$ be represented by $\{E_k\}_k$, then the data encoding block $R(\vect x)$ can be smoothed using the noise channel with Kraus representation $\{V E_k V^H\}_k$.

\paragraph{Nonlinear features} Consider the encoding scheme presented in \citet{quantum_enhanced_fs}, certifying such schemes would require smoothing of the encoding layer of type $RZ(x_1x_2)$. It can be done using data-dependent phase damping layer $\mathcal{N}_{PD}(\lambda)$, with 
\begin{equation}
    \label{remark_eq_nonlinear}
    \lambda = 1 - (\phi(x_1) \phi(x_2) \mathbb{E}_{\substack{\delta_1 \sim \phi \\ \delta_2 \sim \phi}}[\exp(-\iota \delta_1 \delta_2)])^2.
\end{equation} 
Note that \autoref{remark_eq_nonlinear} assumes that each dimension is smoothed using the i.i.d. distribution $\phi$. Furthermore, for Gaussian smoothing $\lambda = 1 - \frac{\exp(-\sigma^2(\vect x_1^2 + \vect x_2^2))}{1 + \sigma^4}$.

\paragraph{Guarantees vs Computational cost} 
Notably, from Corollary \ref{cor:sequential_gaussian}, if we break down a $d$ qubit parallel layer as a stack of $L$ smaller sequential encoding layers the certifiable guarantees decrease as $\mathcal{O}(\frac{1}{\sqrt{L}})$. On the other hand we have to perform eigenvalue decomposition of $L$ matrices of dimension $2^{\frac{d}{L}} \times 2^{\frac{d}{L}}$, thus reducing the computational cost.
\label{sec:methods}

\section{Experiments}
In the following, we conduct an empirical evaluation of our proposed certificate and examine the susceptibility of the smoothed models to gradient-based attacks. To evaluate and compare against formal certificates, we use Projected Gradient Decent (PGD) to compute empirical robustness~\cite{madry2019deep}. In a PGD attack, we try to find an instance in the desired neighborhood $\mathcal B_{\vert \lvert \cdot \rvert \rvert, \epsilon}(\vect x)$ that maximizes some loss. We use gradient descent to solve the maximization of the Binary Cross Entropy (BCE) loss and after every iteration, we project the output back to $\mathcal B_{\vert \lvert \cdot \rvert \rvert, \epsilon}(\vect x)$. We analyze three datasets: (a) the TwoMoons dataset from \citet{scikit-learn}, (b) an artificially generated Annular dataset with known ground truth, and (c) the MNIST handwritten digits \cite{deng2012mnist}. 

We use the strategy proposed by \citet{expo_growing} to obtain a Fourier series representation of the data with the spectrum size increasing exponentially with the qubits. Within each data encoding layer ${L_i}$ with $N$ qubits, we encode a single feature utilizing the gates $\{RZ(2^k x_{i})\}_{k \in \{0,1, \ldots, N-1\}}$. We refer to such an encoding layer as exponential. On the contrary, if we use the gates $RZ(x_{i})$ throughout the layer, we refer to it as linear. We use two strategies to smooth an exponential layer, namely \textit{exponential} and \textit{uniform} smoothing. In exponential smoothing, we smooth $RZ(2^k x_{i})$ with respect to $x_i$ by applying the phase damping channel $\mathcal N_{PD}(1 - \exp(-2^{2k}\sigma^2))$. For uniform, we smooth the encoding $RZ(2^k x_{i})$ with respect to $2^k x_i$, using the channel $\mathcal N_{PD}(1 - \exp(- \sigma^2))$. Let a variational QML model encode every feature $x_i$ using $L$ exponential layers with $N$ qubits each and $\underline p$ ($\overline{p}$) be the lower (upper) bound of the probability $p$ that the classifier predicts class $1$, then using the arguments of Theorem \ref{Th:Sequential_Encoding} we have the following guarantees:

\noindent\textbf{Exponential Smoothing:}
The model remains robust for perturbation less than $\lvert \lvert \matr \delta \rvert \rvert_2 < \frac{\sigma}{\sqrt{L}}\Phi^{-1}(\underline p)$ ($\lvert \lvert \matr \delta \rvert \rvert_2 < \frac{\sigma}{\sqrt{L}}\Phi^{-1}(1 - \overline p)$ if $p < 0.5$).

\noindent\textbf{Uniform Smoothing:}
The model remains robust for perturbation less than $\lvert \lvert \matr \delta \rvert \rvert_2 < \frac{\sigma}{\sqrt{4^{N-1} \frac{L}{3}}}\Phi^{-1}(\underline p)$ ($\lvert \lvert \matr \delta \rvert \rvert_2 < \frac{\sigma}{\sqrt{4^{N-1} \frac{L}{3}}}\Phi^{-1}(1 - \overline p)$ if $p < 0.5$).

In our evaluations of the smooth classifier through the addition of appropriate noise channels, we perform multiple runs of the QNN and use the $95\%$ confidence Clopper-Pearson interval to ascertain the lower bound of the smooth classifier\footnote{Although we realize smoothing as a quantum operation; QML models are still stochastic.}. 

\subsection{TwoMoons Dataset}
We first evaluate our framework on the TwoMoons dataset from Scikit-Learn \cite{scikit-learn}  comprising $250$ data points, which we partition into an $80-20$ train-test split. Subsequently, we train a variational quantum model, using exponential parallel layers ${L_i}, {i\in \{0,1,2,3\}}$. We encode a single feature using the gates $\{RZ(2^k x_{i\%2})\}_{k \in \{0,1,2,3\}}$. We employ the TwoLocals layers as intermediate parameterized layers, followed by the Real Amplitude layers acting as the final parameterized layers\footnote{Refer to Qiskit's documentation \cite{Qiskit} for RealAmplitudes and TwoLocal layers}. We train the model using BCE loss to achieve close $100\%$ accuracy on the test dataset. For smoothing, we consider both the exponential and the uniform strategies. 

Although the potential guarantees exponentially decrease with qubits in a single layer ($N$) for uniform smoothing, for $N \approx 4$ we observe that the guarantees obtained in this case outperform the guarantees for exponential smoothing (see \autoref{fig:certi_acc_moon}). This could be due to the removal of high frequencies in the exponential smoothing case, which also impacts the learning. Note that we need to retrain the exponentially smoothed models to achieve comparable performance, while the uniformly smoothed models do not require retraining for $\sigma \leq 0.75$. This could be attributed to our observation in \autoref{fig:kernels}. 
We observe that performance (validation accuracy and loss) decreases significantly with increasing $\sigma$ for exponential smoothing compared to uniform smoothing. Similar trends are observed in attacking the smooth classifiers as depicted in \autoref{fig:grad_attack_moon}. 

\begin{figure}[h]
\centering
\begin{subfigure}
  \centering
  \includegraphics[width=0.8\linewidth]{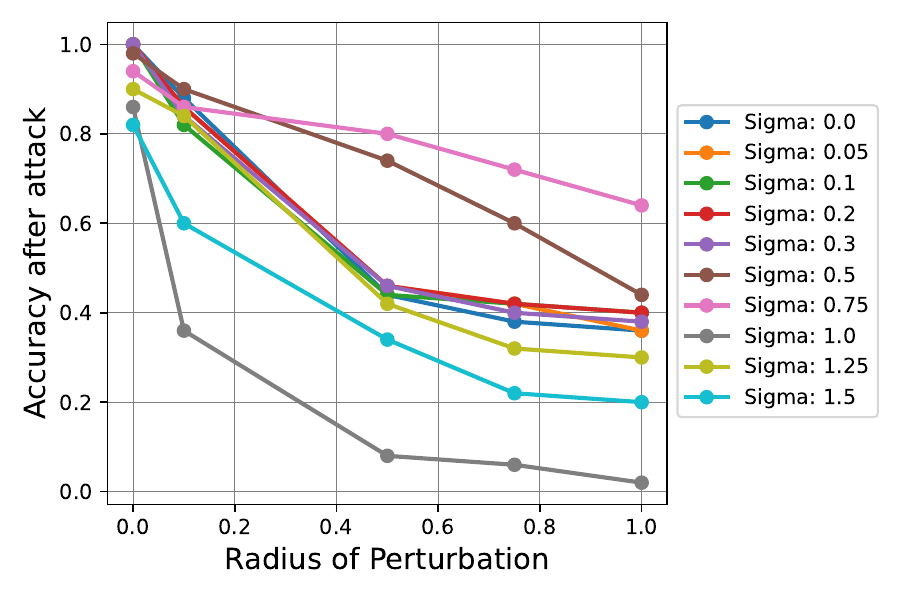}
\end{subfigure}
\begin{subfigure}
  \centering
  \includegraphics[width=0.8\linewidth]{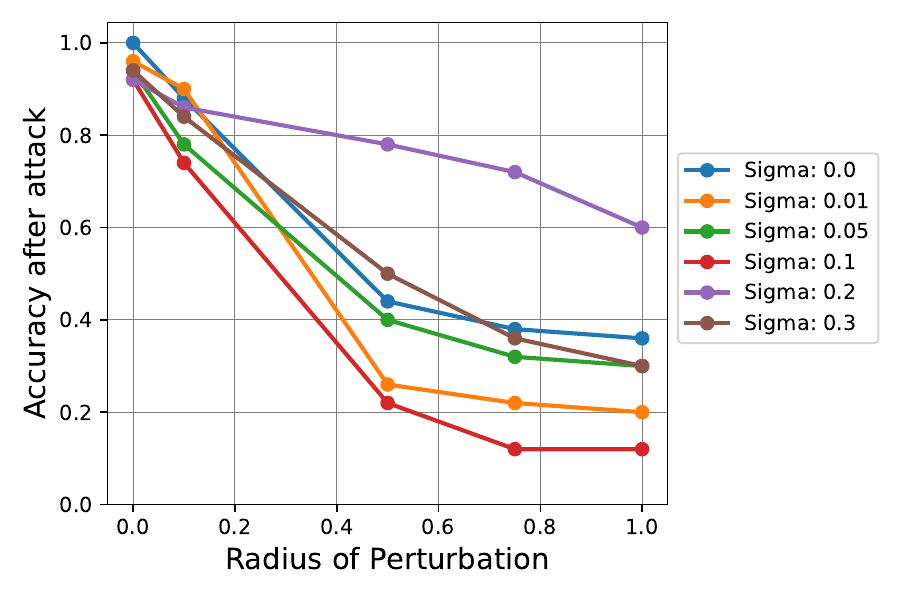}
\end{subfigure}
\caption{Accuracy upon gradient-based attacks for smooth classifiers on TwoMoons Dataset with different values of $\sigma$. (a) All data encoding blocks are smoothed using the Phase-Damping noise channels with parameter $\lambda = 1 - \exp(-\sigma^2)$. (b) All data encoding blocks encoding $\alpha x$ are smoothed using Phase-Damping noise channels with parameter $\lambda = 1 - \exp(-\alpha^2 \sigma^2)$.}
\label{fig:grad_attack_moon}
\end{figure}
\begin{figure}[h]
 \centering
\begin{subfigure}
  \centering
  \includegraphics[width=0.8\linewidth]{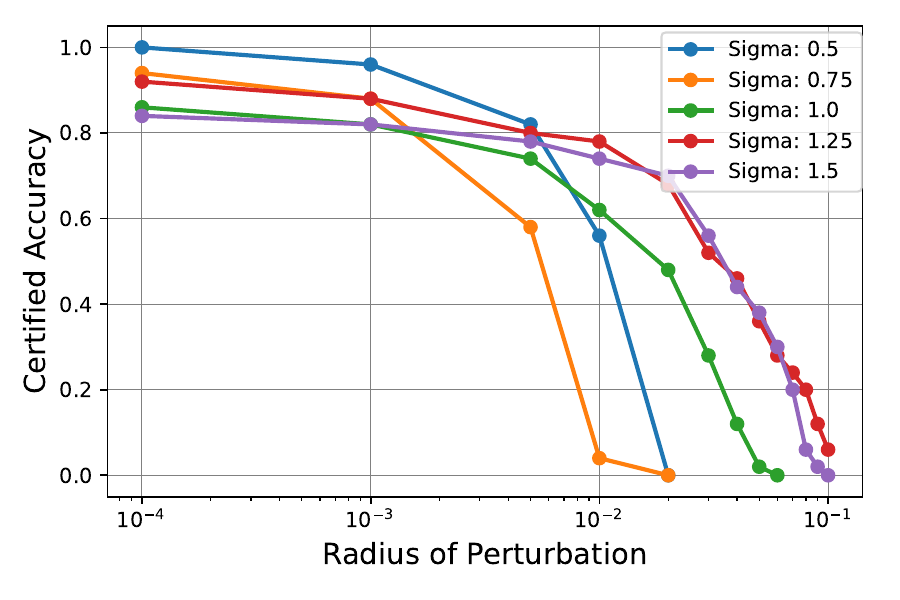}
\end{subfigure}
\begin{subfigure}
  \centering
  \includegraphics[width=0.8\linewidth]{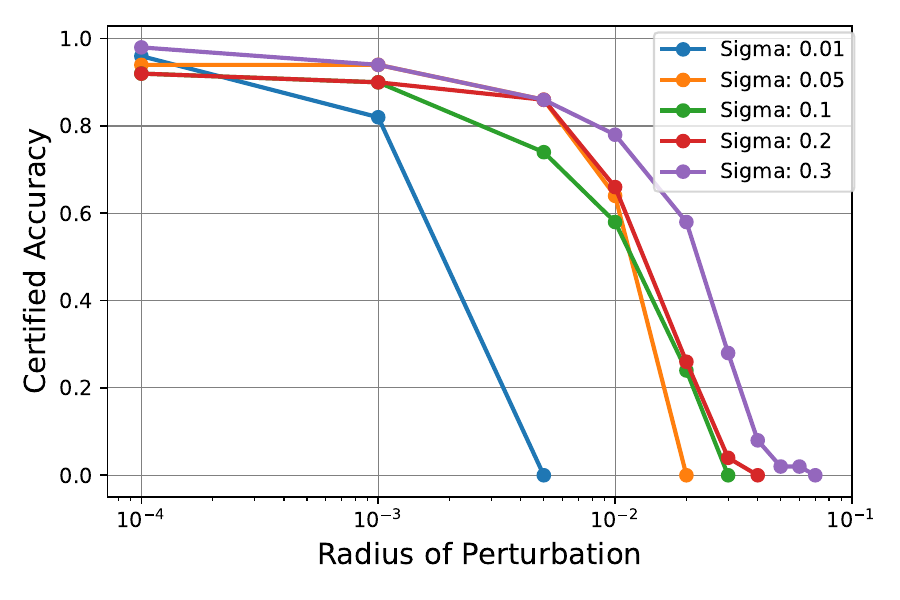}
\end{subfigure}
\caption{Fraction of test dataset correctly classified and certified against the radius of perturbation for smooth classifiers on TwoMoons Dataset with different values of $\sigma$. (a) All data encoding blocks are smoothed using phase damping noise channels with parameter $\lambda = 1 - \exp(-\sigma^2)$. (b) All data encoding blocks encoding $\alpha x$ are smoothed using phase damping noise channels with parameter $\lambda = 1 - \exp(-\alpha^2 \sigma^2)$.}
\label{fig:certi_acc_moon}
\end{figure}

\subsection{Annular Dataset}
We proceed by evaluating our framework on an artificially constructed Annular dataset based on the following ground truth $y(\vect x) = 1 - \mathbbm 1[0.3 < \lvert\lvert \vect x \rvert\rvert_2 \le 0.8]$. Access to the ground truth allows us to check if the semantic meaning of the data is preserved when attacking with gradient-based methods.  We only consider an attack successful if the prediction of the model changes, but the semantics remain the same. In this experiment, we do kernel-based training for our model. Similarly to \citet{gaussian_pennylane}, we approximate the $2-D$ Gaussian kernel using two layers of exponential encoding with $10$ qubits. To speed up the experiment, we derive the symbolic representation of the Fourier series with exponential and uniform smoothing and use that to evaluate the kernel and its gradient w.r.t. the data. We plot the accuracy after attacking the uniformly and exponentially smoothed models against the radius of the attack in Figure \ref{fig:grad_attack_annular}.

This toy example allows us to study the behavior of the exponential and the uniform smoothing that we observed in previous experiments (Figure \ref{fig:grad_attack_moon}, \ref{fig:certi_acc_moon}). The uniformly smoothed models are easier to train and have better provable and empirical robustness. We plot the rescaled kernel evaluations obtained from the quantum circuit without smoothing, with uniform smoothing and exponential smoothing in Figure $\ref{fig:kernels}$. In general, we observe that for higher values of $\sigma$ the shape of the uniform encoding remains closer to the actual non-smooth kernel while exponential smoothing seemingly penalizes high frequencies more. We attribute this to be the reason for the good performance observed in the case of uniform smoothing even for larger values of $\sigma$. 

\begin{figure}[h]
 \centering
\begin{subfigure}
  \centering
  \includegraphics[width=0.8\linewidth]{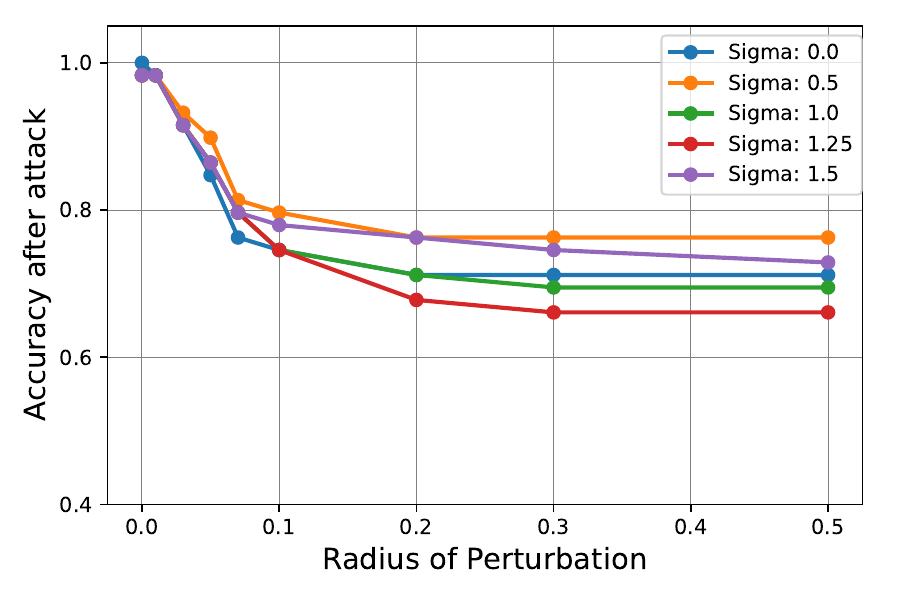}
\end{subfigure}
\begin{subfigure}
  \centering
  \includegraphics[width=0.8\linewidth]{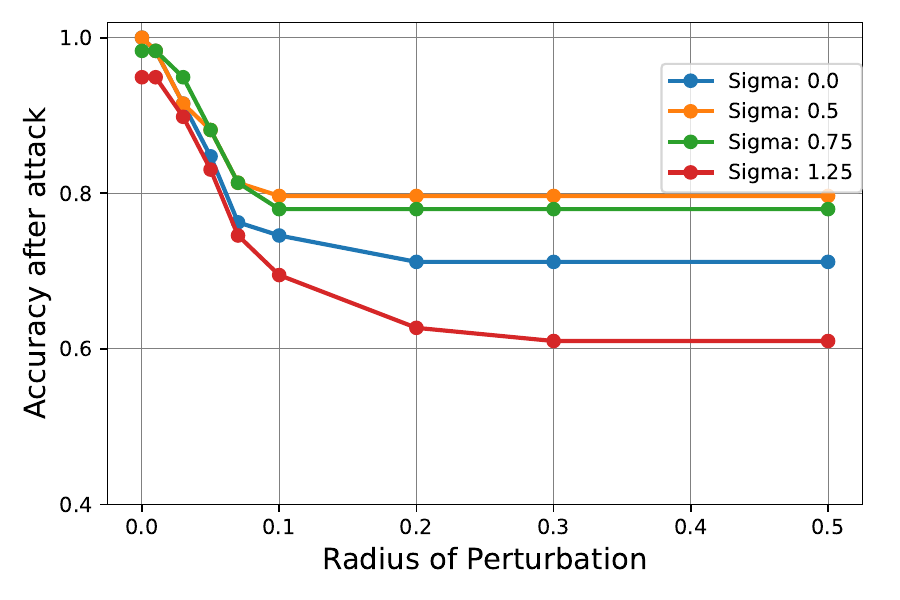}
\end{subfigure}
\caption{Accuracy upon gradient-based attacks for smooth classifiers on artificially generated Annular Dataset with different values of $\sigma$. (a) All data encoding blocks are smoothed using phase damping noise channels with parameter $\lambda = 1 - \exp(-\sigma^2)$. (b) All data encoding blocks encoding $\alpha x$ are smoothed using Phase-Damping noise channels with parameter $\lambda = 1 - \exp(-\alpha^2 \sigma^2)$.}
\label{fig:grad_attack_annular}
\end{figure}
\begin{figure}[h]
\centering
  \centering
  \includegraphics[width=.8
\linewidth]{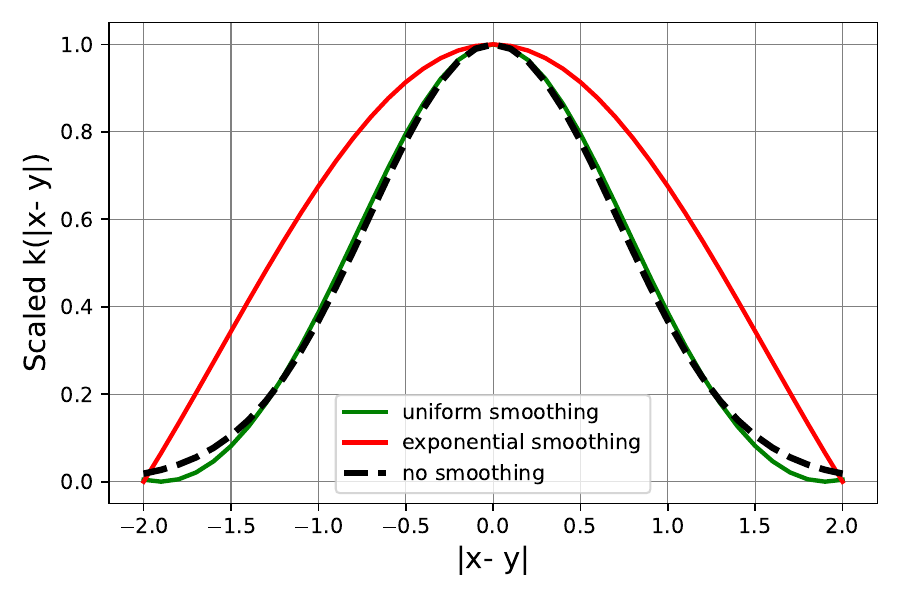}
\caption{Kernel functions obtained from the Quantum circuit without smoothing(approximation of Gaussian), with exponential smoothing and uniform smoothing for $\sigma = 1.5$. The smooth kernels are scaled back to 0, 1.}
\label{fig:kernels}
\end{figure}

\subsection{MNIST}
We tackle a binary classification problem using the $28 \times 28$ MNIST \cite{deng2012mnist} handwritten digits dataset, focusing solely on images of digits $0$ and $1$. We adopted a hybrid classical-quantum model for classification. We convert a $784$-dimensional input data into a $6$-dimensional intermediate representation by employing a trainable linear layer and further process it using a QML model.
For the quantum component, we use a $3$-qubit system and employ $4$ data encoding layers that encode the data using $RZ$ gates. The first and third layers encoded the first three dimensions, whereas the second and fourth layers encoded the remaining three dimensions of the intermediate state. Each encoded layer was preceded and succeeded by a parameterized TwoLocal circuit. Subsequently, after the data encoding block, we apply a RealAmplitudes parameterized layer. Finally, all qubits were measured and the probability of an odd number of $1's$ in the measured binary string was interpreted as the likelihood of predicting class $1$ (that is, the digit $1$).

We randomly selected 200 samples from the training split of the MNIST dataset, ensuring an equal number of samples from each class. The model was trained to minimize the BCE Loss. For model evaluation, we randomly selected $1000$ images ($500$ from each class) from the validation split. We used phase damping noise channels, $\mathcal N_{PD}(1 - \exp(-\sigma^2))$ to smooth the intermediate representation using a zero-centered Gaussian distribution with $\sigma$ as the standard deviation. We categorized our results into the following two segments: \textit{formal guarantees} and \textit{empirical attack-based} evaluation.\\

\noindent\textbf{Formal Guarantees:}
To certify for perturbations in the input images, we first employ our framework to guarantee robustness in the intermediate states. Then we evaluate the spectral norm of the classical linear layer to extend the guarantees to the perturbation in the actual images. For different values of $\sigma$, we show the variation of the certified ratio, the certified accuracy with the strength of the potential attack, using the $L2$ norm in Figure \ref{fig:MNIST_SMOOTH_CERTI_RATIO_acc}. We also plot the variation of accuracy with different values of $\sigma$.

The accuracy remains nearly $100\%$ for all values of $\sigma$ between $0$ and $1$, therefore the plots for the certified ratio and the certified accuracy exhibit similar trends. A decrease in accuracy with the increase in $\sigma$ is reflected in a decrease in certified accuracy.
Furthermore, it should be noted that the peak in the certified ratio is for $\sigma \in \{0.25, 0.5\}$. With an additional increase in $\sigma$, a large fraction of images cannot be certified. This phenomenon may be attributed to over-smoothing, leading to a greater number of points being situated close to the decision boundary.
\begin{figure}[h]
\begin{subfigure}
  \centering
  \includegraphics[width=0.95\linewidth]{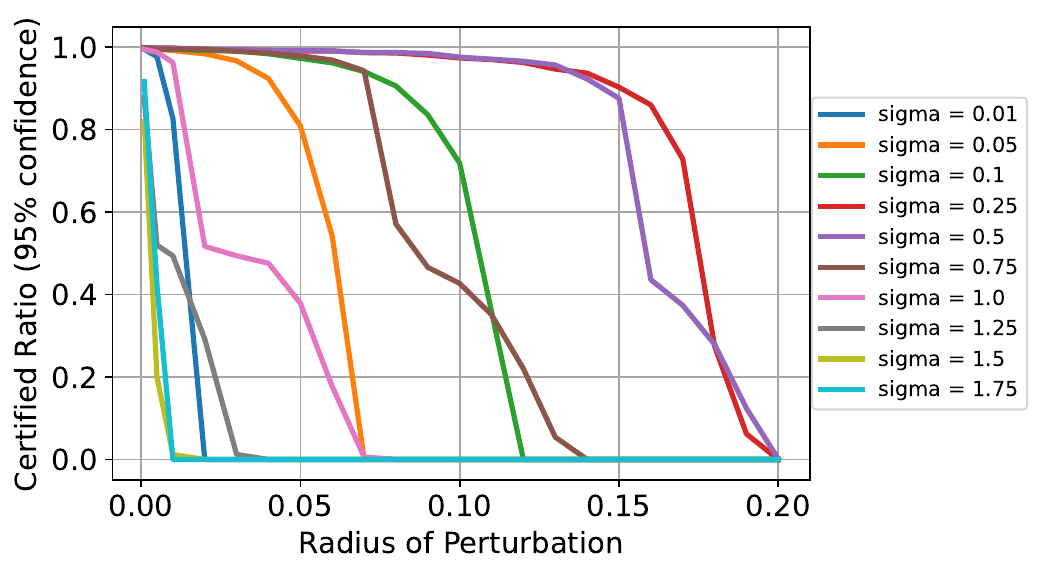}
\end{subfigure}%
\begin{subfigure}
  \centering
  \includegraphics[width=0.95
\linewidth]{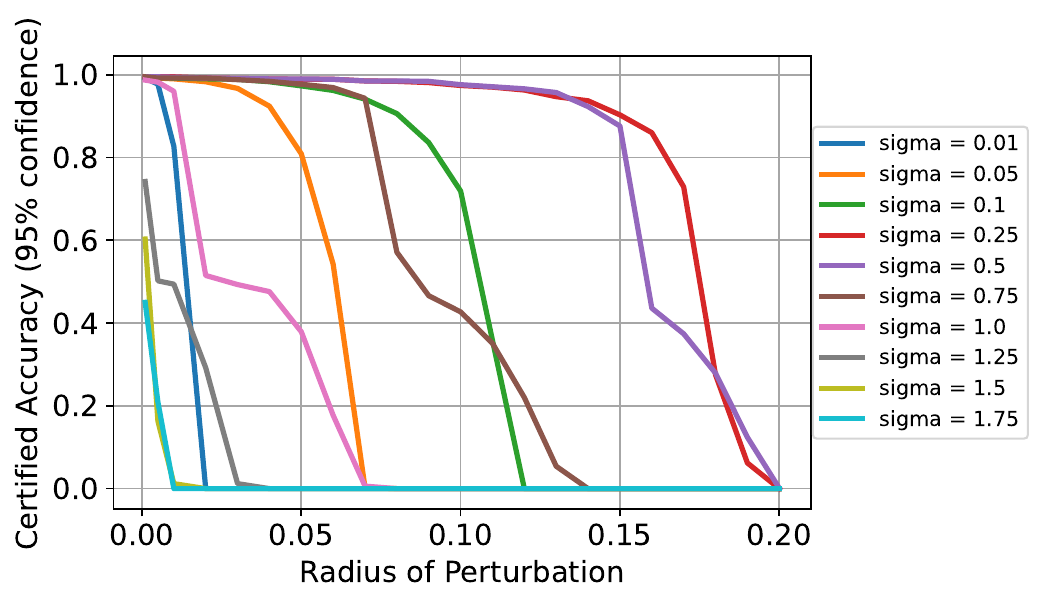}
\end{subfigure}
\begin{subfigure}
  \centering
  \includegraphics[width=0.75
\linewidth]{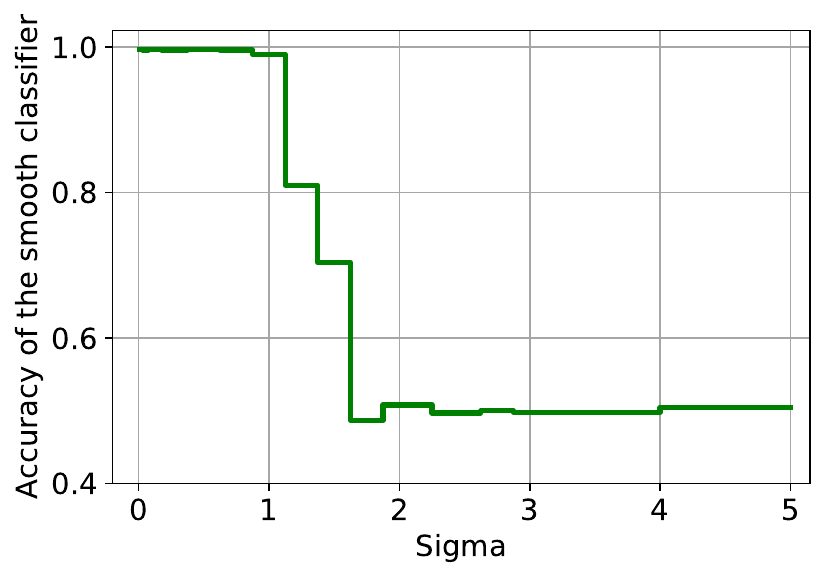}
\end{subfigure}
\caption{(a) Certified ratio and (b) Certified accuracy, against a given radius of perturbation in the $L2$ norm, (c) Accuracy of the smooth classifier, for different values of $\sigma$}
  \label{fig:MNIST_SMOOTH_CERTI_RATIO_acc}
\end{figure}
The critical insight from these plots is that our model achieves accuracy levels comparable to the base classifier within the range of $\sigma$ values between $0.25$ and $0.5$. Remarkably, almost $90\%$ of the images demonstrate certifiable robustness to perturbations within an $L_2$ bound of $0.15$. Notably, achieving these guarantees merely entails the addition of appropriate noise channels\footnote{To enhance the performance of the smooth classifier one could also fine-tune the smooth circuit by training it on the training dataset}, providing the user with an increase certified robustness while having no decrease in predictive performance.

\noindent\textbf{Projected Gradient Decent based attacks:}
Given the evaluation of the formal guarantees, we next assess the empirical increase in the robustness of our approach. We compare the impact of gradient-based attacks on the accuracy of the base classifier with the certified accuracies of the smooth classifier for values of $\sigma$ near $0.25$ and $0.5$. The accuracy under attack serves as an upper bound to the actual certified accuracy, acting as a surrogate and an upper bound for the actual certification in our guarantee evaluation. We show these comparisons in Figure \ref{fig:grad_attack_AND_ACC}. The gap between the certification and the attack is very significant. 
\begin{figure}[t]
\centering
\begin{subfigure}
  \centering
  \includegraphics[width=0.8\linewidth]{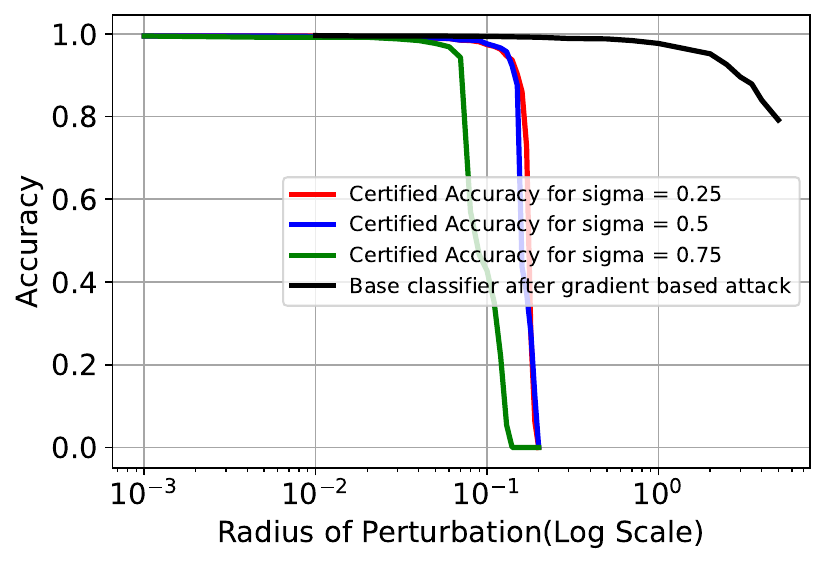}
\end{subfigure}

\begin{subfigure}
  \centering
  \includegraphics[width=0.8\linewidth]{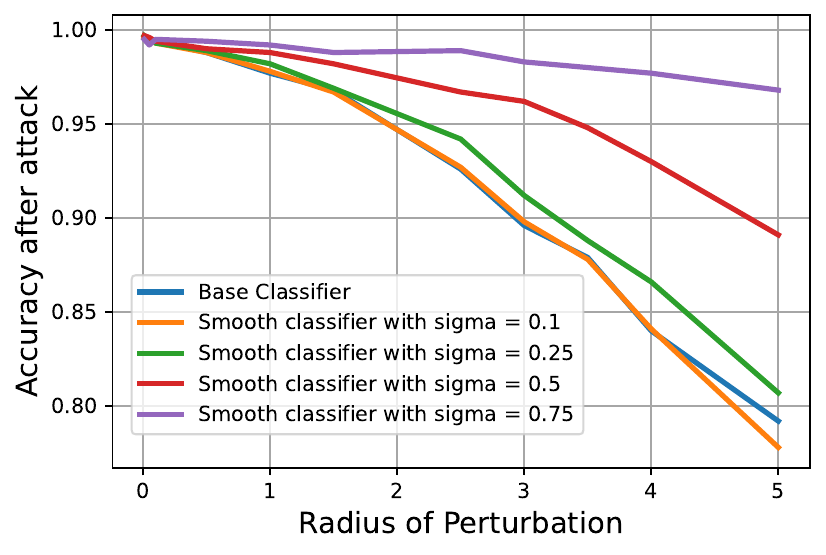}
\end{subfigure}%
\caption{(a) Gap between the certified accuracy and accuracy of the base classifier upon gradient-based attacks. (b) Accuracy upon gradient-based attacks for smooth classifiers with different values of $\sigma$}
\label{fig:grad_attack_AND_ACC}
\end{figure}
Secondly, in Figure \ref{fig:grad_attack_AND_ACC}, we evaluate the effect of these attacks on the smooth classifier for different values of $\sigma$. We see that increasing $\sigma$ from $0.1$ to $0.75$, makes it difficult for the PGD attack to influence the model's prediction. In addition, it highlights that the gap between the upper bound (attack-based accuracy) and the lower bound (certification using our approach) is even greater. To illustrate this, observe the disparity between the attacked base classifier and the certified smooth classifier for $\sigma = 0.75$ in Figure \ref{fig:grad_attack_AND_ACC} (a) and further contrast between the attacked base classifier with the attacked smooth classifier for $\sigma = 0.75$ in Figure \ref{fig:grad_attack_AND_ACC} (b).

\label{sec:Experiments}

\section{Conclusion}
Our work is the first to certify QML models with near-term general encoding schemes against adversaries manipulating data in the classical domain. We demonstrate that smoothing a Hamiltonian-based encoding into parallel subsystems can be mathematically represented as a quantum noise channel. Sequentially stacking these smooth parallel layers certifies the robustness of quantum models. However, our extensive experimental evaluation reveals a significant gap between what we can certify and the influence of gradient-based attacks on the model's prediction.

Future work should investigate improving robustness guarantees by incorporating details of the underlying base classifier. Additionally, empirically studying the impact of added noise on the training dynamics and generalization of QML models is another important direction. Nevertheless, this work lays a foundational step toward using classical robustness certification techniques in QML.
\label{sec:conclusion}

\appendix
We start with stating the following Lemma suggesting that the characteristic function defines a valid kernel. %
    \begin{lemma}
    \label{lemma:phikernel}
        Let $(\mathbb R, \Sigma, \mu)$ define a Borel probability space with $\phi$ as the characteristic function, then $\kappa: \mathbb R \times \mathbb R \mapsto \mathbb C$; $\kappa(x,y) := \phi(x-y)$ is a kernel function.
     \end{lemma}
     \begin{proof}
         First note that $\kappa(x,y):= \int \exp(\iota(x-y)z) d\mu(z)$ is anti-symmetric in the arguments, thus the resulting Gram matrix will be hermitian, therefore we just need to show that for any Gram matrix $\matr A \in \mathbb C^{n \times n}; A_{ij}: = \kappa(x_i, x_j)$, $\forall \vect u \in \mathbb C^n$, $\vect u^H \matr A \vect u \ge 0$. Define $\vect a(z): \mathbb R \mapsto \mathbb C ^ n$ as $a_i(z): = \exp(\iota x_i z)$, then $\vect u^H \matr A \vect u = \int \vect u^H  \vect a(z) \vect a(z)^H \vect u d \mu(z) = \int \lvert \vect a(z)^H \vect u \rvert ^2 d \mu(z) \geq 0$.
     \end{proof}
\begin{remark}
         If the underlying distribution is a zero-centered normal distribution with variance $\frac{1}{\sigma}$, the resulting kernel is an RBF kernel defined as $\kappa_{RBF} (x,y):= \exp(\frac{-\lvert x -y \rvert ^2}{2\sigma^2})$. 
          \label{Re:RBF}
\end{remark}
\begin{proof}[\textbf{Theorem \ref{Th:Parallel_Encoding}: Proof}]
    \begin{enumerate}[label=(\alph*)]
    \label{proof:Parallel}
        \item This part has two statements, we have already discussed the latter while the former follows from (b), so essentially if we prove (b), (a) follows. Nevertheless, (a) can also be shown by explicitly showing the map is completely positive and trace-preserving. A crucial step in showing that is the following observation about the given point-wise multiplication, for any given $\rho$, $\mathcal E_\phi(\rho) = \mathbb E_{z\sim \phi}[U(z)\rho U(z)^H]$. This representation directly allows us to check for the fact that $\mathcal E _\phi$ is trace-preserving and positive. To show complete positivity introduce a system \textbf{R} of arbitrary dimensions to our system \textbf{Q}. We need to show that $\mathbb I_{\textbf R} \otimes \mathcal E _\phi$ is positive. Any positive operator $\rho_{\textbf{RQ}}$ on the system \textbf{RQ} can be decomposed as $\sum_k \sigma_k \ket{\psi^{(k)}_{\textbf{RQ}}} \bra{\psi^{(k)}_{\textbf{RQ}}}$ with $\sigma_k \ge 0$. Therefore showing $\mathbb I _ \textbf R \otimes \mathcal E _\phi(\ket{\psi_{\textbf{RQ}}} \bra{\psi_{\textbf{RQ}}}) \succeq 0 $ for any pure quantum state $\ket{\psi_{\textbf{RQ}}}$ allows us to establish the positivity of $\mathbb I_{\textbf R} \otimes \mathcal E _\phi$. To show that use the Schmidt decomposition of $\ket{\psi_{\textbf{RQ}}}$ as $\ket{\psi_{\textbf{RQ}}} = \sum_i \alpha_i \ket{i_{\textbf{R}}} \ket{i_{\textbf{Q}}}$, then $\forall \ket{u_\textbf{RQ}}$,
        \begin{equation}
        \label{eq:Complete_positivity}
        \begin{split}   
        &\bra{u_\textbf{RQ}}(\mathbb I _ \textbf R \otimes \mathcal E _\phi(\ket{\psi_{\textbf{RQ}}} \bra{\psi_{\textbf{RQ}}})) \ket{u_\textbf{RQ}} \\
        &= \bra{u_\textbf{RQ}}(\Sigma_{ij} \alpha_i \alpha_j \ket{i_\textbf R}\bra{j_\textbf R} \otimes \mathbb E_{z\sim \phi}[U(z)\ket{i_\textbf Q}\bra{j_\textbf Q}\\
        & \quad \quad U(z)^H]) \ket{u_\textbf{RQ}} \\
        &= \mathbb E_{z\sim \phi}[\bra{u_\textbf{RQ}}\mathbb I _ \textbf R \otimes U(z)(\Sigma_{ij} \alpha_i \alpha_j \ket{i_\textbf R}\bra{j_\textbf R} \otimes \ket{i_\textbf Q}\bra{j_\textbf Q})\\
        & \quad \quad \mathbb I _ \textbf R \otimes U(z)^H \ket{u_\textbf{RQ}}] \\
        &= \mathbb E_{z\sim \phi}[\bra{u_\textbf{RQ}}\mathbb I _ \textbf R \otimes U(z)\ket{\psi_\textbf{RQ}} \bra{\psi_\textbf{RQ}}\mathbb I _ \textbf R \otimes U(z)^H \ket{u_\textbf{RQ}}] \\
        &= \mathbb E_{z\sim \phi}[\lvert \bra{\psi_\textbf{RQ}}\mathbb I _ \textbf R \otimes U(z)^H \ket{u_\textbf{RQ}} \rvert ^ 2] \ge 0\\
        \end{split}
        \end{equation}        
        \item Positivity follows from Lemma \ref{lemma:phikernel}, while the second part can be shown as follows:
        \begin{equation}
        \label{eq:kraus}
        \begin{split}   
        \sum_k E_k \rho(x) E_k^H &= \sum_k \sigma_k \rho(x) \odot u_k u_k^H \\
        &= \rho(x) \odot \sum_k \sigma_k u_k u_k^H \\
        &= \rho(x) \odot A_{\phi, \lambda} 
        \end{split}
        \end{equation}
        
        Equation \ref{eq:kraus} also allows us to show that $\sum_k E_k^H E_k = \sum_k E_k E_k^H = \sum_k E_k \mathbb{I} E_k^H  = \mathbb{I} \odot A = \mathbb{I}$. Thus $\{E_k\}_k$ gives the Kraus representation of the desired transformation. Since the Kraus representation exists the concerned operator is CPTP thus implying (a). Note that the above representation only holds for $\sigma_k \ge 0$, therefore $A_{\phi, \lambda} \succeq 0$ was crucial in verifying the Kraus representation.%
        \item To find the upper bound of the trace distance we maximize over all the $\rho \in \mathcal{M}_d := \{\rho |\rho : \mathbb R \mapsto \mathcal S _+^{d(1)}\}$ and to solve the optimization problem further we write the trace distance as a semi-definite program.
        \begin{align*}
            &d(\mathcal E_\phi(\rho(x)), \mathcal E_\phi(\rho(y))) \\
            &\leq \max_{\sigma \in \mathcal M_d} d(\mathbb E _{\delta \sim \phi} [\sigma(x + \delta)], \mathbb E _{\delta \sim  \phi} [\sigma(y + \delta)])\\
            &\leq \max_{0\preceq P \preceq \mathbb I_{2^d} } \max_{\sigma \in \mathcal M_d}\\& \quad\quad\quad Tr(P(\mathbb E _{\delta \sim  \phi} [\sigma(x + \delta)] - \mathbb E _{\delta \sim  \phi} [\sigma(y + \delta)])) \\
            &= \max_{\substack{0\preceq P \preceq \mathbb I_{2^d} \\ \sigma \in \mathcal M_d} } \int_{\mathbb R} Tr(P\sigma(z)) (f_\phi(z-x) - f_\phi(z-y)) d\lambda(z) \\
            &\leq \max_{\substack{f \in \mathcal{F}_{[0,1]} \\:= \{f | f:\mathbb R \mapsto [0,1]\}}} \int_{\mathbb R} f(z) (f_\phi(z-x) - f_\phi(z-y)) d\lambda(z)\\
            &= \int_{\mathbb R} \max(f_\phi(z-x) - f_\phi(z-y), 0) d\lambda(z)
        \end{align*}
        The last inequality follows from the fact that $\{Tr(P\sigma)| 0\preceq P \preceq \mathbb I_{2^d} \land \sigma \in \mathcal M_d\} \subseteq \mathcal{F}_{[0,1]}$ and therefore the maximum is obtained by $f(z)= 1$ if $\mu_x(z) \geq \mu_y(z)$ and $0$ otherwise
        Additionally, note that $\int_{\mathbb R} \max(f_\phi(z-x) - f_\phi(z-y), 0) d\lambda(z) = \int_{\mathbb R} \max(f_\phi(z-y) - f_\phi(z-x), 0) d\lambda(z)$.%

    \end{enumerate}
\end{proof}

\begin{proof}[\textbf{Corollary \ref{cor:Parallel_Gaussian}: Proof}]
\begin{enumerate}[label=(\alph*)]
    \item Let $\mu_x$ and $\mu_y$ be the PDF of the Gaussian distributions centered at $x$ and $y$ respectively and define $C:= \{r\in \mathbb R | \mu_x(z) - \mu_y \geq 0\} = \{r\in \mathbb R | (y-x)(z) \leq \frac{y^2 - x^2}{2}\}$
    \begin{align*}
        &\int_{\mathbb R} max(\mu_x(z) - \mu_y(z), 0) d\lambda(z)= \int_{C} \mu_x(z) - \mu_y(z) d\lambda(z) \\
        &= \int_{C} \mu_x(z) dz - \int_{C} \mu_y(z) d\lambda(z) \\
        &= \mathbb P[z \sim \mathcal{N}(x(y-x), (|y-x|\sigma)^2) \leq \frac{y^2 - x^2}{2}] - \\
        &  \quad \mathbb P[z \sim \mathcal{N}(y(y-x), (|y-x|\sigma)^2) \leq \frac{y^2 - x^2}{2}] \\
        &= \Phi(\frac{|y-x|}{2\sigma}) - \Phi(-\frac{|y-x|}{2\sigma}) 
        = 2(\Phi(\frac{|y-x|}{2\sigma}) - \frac{1}{2})
    \end{align*}
\item Due to the linearity of the quantum operations regarding the input density matrices we can write $f(x)$ as: 
\begin{align*}
    f(x) &= Tr(P\mathcal E (\mathbb E_\delta[\rho(x + \delta)])) \\
&= \mathbb E_\delta [Tr(P\mathcal E (\rho(x + \delta)))] \\
&= \mathbb E_\delta[g(x + \delta)]
\end{align*}
where $g(x)$ is the quantum classifier without the smoothing operation, i.e., $g(x) := Tr(P\mathcal E (\rho(x)))$. Therefore adding the smoothing circuit is equivalent to smoothing the classier $g(x)$ using Randomized smoothing, hence using the robustness guarantee from \cite{cohen2019} shows the claim made in (b).
\end{enumerate}
\end{proof}
\begin{proof}[\textbf{Theorem \ref{Th:Sequential_Encoding}: Proof}]
        \begin{enumerate}[label=(\alph*)]
    \item In order to prove the first part we will use mathematical induction. We start with defining the induction hypothesis $S(l):= \matr{\tilde \rho}_{l, \phi}(x) = \mathbb E_{\svect {\delta^{(l)}}}[\sigma_l(\vect {x^{(l)}}+ \svect {\delta^{(l)}})]$, i.e. the claim holds for sequential encoding with $l$ parallel layers. Now the base case follows from the results for the parallel encoding scheme, i.e. Theorem \ref{Th:Parallel_Encoding}. To prove the induction step assume that the result holds for $l$ layers, i.e. $S(l)$ holds, then using Fubini's theorem(for the second last equality) we have the following, 
    \begin{align*}
        &\tilde\rho_{l+1, \phi}(x) = \mathcal E^{l+1}_\phi(U_{l+1}(x) W^{(l+1)} \matr{\tilde \rho}_{l, \phi}(x)U_{l+1}(x)^H W^{(l+1)H})\\
        &= \mathbb E_{\delta_{l+1}}[U_{l+1}(x + \delta_{l+1}) W^{(l+1)} \matr{\tilde \rho}_{l, \phi}(x)\\
        & \quad \quad U_{l+1}(x + \delta_{l+1})^H W^{(l+1)H}] \\
        &= \mathbb E_{\delta_{l+1}}[U_{l+1}(x + \delta_{l+1}) W^{(l+1)} \mathbb E_{\svect {\delta^{(l)}}}[\sigma_l(\vect{x^{(l)}}+ \bm {\delta^{(l)}})] \\
        &\quad \quad U_{l+1}(x + \bm_{l+1})^H W^{(l+1)H}] \quad \text{(Hypothesis)}\\
        &= \mathbb E_{\delta_{l+1}, \svect {\delta^{(l)}}}[U_{l+1}(x + \delta_{l+1}) W^{(l+1)} \sigma_l(\vect{x^{(l)}}+ \svect {\delta^{(l)}}) \\
        & \quad \quad U_{l+1}(x + \delta_{l+1})^H W^{(l+1)H}] \quad \text{(Fubini)}\\
        &= \mathbb E_{\svect {\delta^{(l+1)}}}[\sigma_{l+1}(\vect{x^{(l+1)}}+ \svect {\delta^{(l+1)}})] \\
    \end{align*}

    Therefore, $S(l)$ holds for all $l \in \mathbb N$ and hence the claim in the statement is established.

    \label{sub_proof:sequential_smoothing}

    \item Observe that using the results from \ref{sub_proof:sequential_smoothing} with an additional assumption of the existence of probability density function $f_\phi$ one can write the following, 
    \begin{align*}
        \matr{\tilde \rho}_{L,\phi}(x) = \int_{\mathbb R ^ L} \sigma_L(\vect z) \prod_{i=1}^{L} f_\phi(z_i - x) d\lambda(\vect z)
    \end{align*}
    Now, use the arguments in the proof of Theorem \ref{Th:Parallel_Encoding}\ref{subth:trace_distance} to show the claim.
    \end{enumerate}
\end{proof}

\begin{proof}[\textbf{Theorem \ref{th:singlesmooth}: Proof:}]
    $  $\newline
    Using the arguments from The Proof of Theorem \ref{Th:Parallel_Encoding} (a), we just need to show the statement for $\rho  = \rho_1 \otimes \rho_2$, which further translates to showing that:
    \begin{equation}
       \mathcal N_{PD}(RZ(\vect x) \text{ } \matr \rho \text{ } RZ(\vect x)^H) = \mathbb{E}_{\matr \delta \sim \phi}[RZ(\vect x + \matr \delta) \text{ } \matr \rho \text{ } RZ(\vect x + \matr \delta)^H ] 
    \end{equation} 

    To show that observe, 
        \begin{equation}
        \begin{aligned}
        &\mathbb{E}_{\matr \delta \sim \phi}[RZ(\vect x + \matr \delta) \text{ } \matr \rho \text{ } RZ(\vect x + \matr \delta)^H ]\\
        &= \mathbb{E}_{\matr \delta \sim \phi}[\begin{bmatrix}
        \rho_{11} & \rho_{12} \exp(-\iota(\vect x + \matr \delta))\\
        \rho^*_{12} \exp(\iota(\vect x + \matr \delta)) & \rho_{22}
        \end{bmatrix}] \\
        &= \begin{bmatrix}
        \rho_{11} & \rho_{12} \exp(-\iota\vect x) \phi(1)\\
        \rho^*_{12} \exp(\iota\vect x)\phi(1)^* & \rho_{22}
        \end{bmatrix} \\
        &= \mathcal N_{PD}(RZ(\vect x) \text{ } \matr \rho \text{ } RZ(\vect x)^H)
        \end{aligned}
    \end{equation}
\end{proof}

\printbibliography

\end{document}